\documentclass[letterpaper, 10 pt, journal, twoside]{IEEEtran}

\usepackage{blindtext, graphicx}
\usepackage{cite}
\usepackage{graphicx}
\usepackage{amsmath}
\usepackage{amsfonts}
\usepackage{color}
\usepackage{enumerate}
\usepackage{subcaption}
\usepackage[export]{adjustbox}
\usepackage{caption}
\usepackage{amssymb}
\newcommand{\divides}{\bigm|}
\usepackage{color,soul}
\usepackage[english]{babel}
\usepackage{hyperref}

\newtheorem{theorem}{Theorem}

\newtheorem{assumption}{Assumption}

\newtheorem{example}{Example}

\newtheorem{lemma}{Lemma}

\renewcommand{\vec}[1]{\mathbf{#1}}

\newtheorem{remark}{Remark}
\newtheorem{dis}{Discussion}
\makeatletter
\newcommand*{\rom}[1]{\expandafter\@slowromancap\romannumeral #1@}
\makeatother

\newenvironment{proof}{\textit{Proof:}}

\usepackage{algorithmic}

\newif\ifbio
\biotrue 

\newif\ifLongVersion
\LongVersiontrue  

\usepackage[absolute]{textpos}
\begin{document}
%

\begin{textblock}{15}(1,1)
\textcolor{red}{This paper formally accepted and published on IEEE Control Systems Letters (\href{https://ieeexplore.ieee.org/document/9130738}{Link}). } 
\end{textblock}

\title{Neuro-Adaptive Formation Control and Target Tracking for Nonlinear Multi-Agent \\ Systems with Time-Delay}

\author{Kiarash Aryankia$^{ \href{https://orcid.org/0000-0002-4751-3925}{\includegraphics[scale=.04]{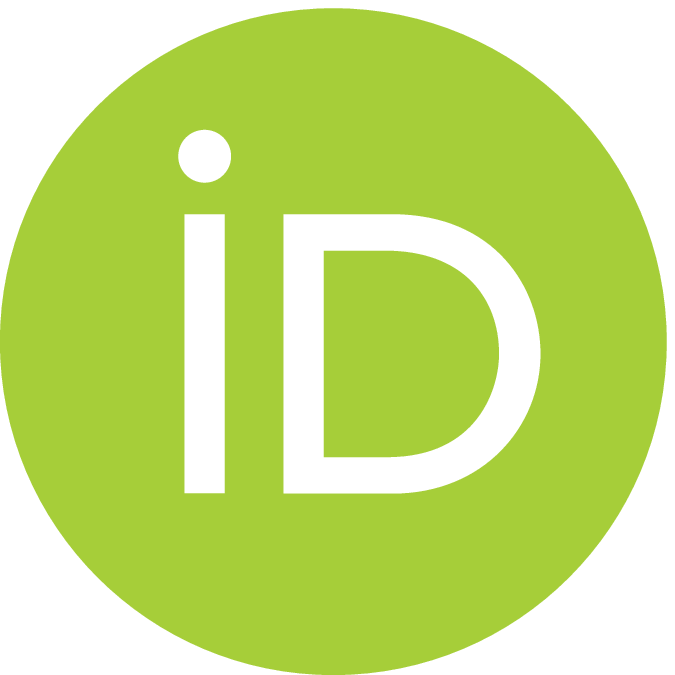}}}$, \IEEEmembership{Student Member, IEEE} and Rastko R. Selmic$^{\href{https://orcid.org/0000-0001-9345-8077}{\includegraphics[scale=.04]{Figures/ORCID-iD_icon-128x128.eps}}}$, \IEEEmembership{Senior Member, IEEE}
\thanks{The authors acknowledge the support of the Natural Sciences and Engineering Research Council of Canada (NSERC), [funding reference \#: RGPIN-2018-05093]. (\textit{Corresponding Author: Kiarash Aryankia}). The authors are with the Department of Electrical and Computer Engineering,
    Concordia University, Montreal, QC, Canada
    {\tt\small (email: k\_aryank@encs.concordia.ca; \tt\small rastko.selmic@concordia.ca).}}}

\maketitle

\begin{abstract}
This paper proposes an adaptive neural network-based backstepping controller that uses rigid graph theory to address the distance-based formation control problem and target tracking for nonlinear multi-agent systems with bounded time-delay and disturbance. The radial basis function neural network (RBFNN) is used to overcome and compensate for the unknown nonlinearity and disturbance in the system dynamics. The effect of the state time-delay of the agents is alleviated by using an appropriate control signal that is designed based on specific Lyapunov function and Young's inequality. The adaptive neural network (NN) weights tuning law is derived using this Lyapunov function. An upper bound for the singular value of the normalized rigidity matrix is introduced, and uniform ultimate boundedness (UUB) of the formation distance error is rigorously proven based on the Lyapunov stability theory. Finally, the performance and effectiveness of the proposed method are validated through the simulation results on nonlinear multi-agent systems. Comparisons between the proposed distance-based method and an existing, displacement-based method are provided to evaluate the performance of the suggested method. 
\end{abstract}
\begin{IEEEkeywords}
  Agents-based systems, delay systems, formation control, neural networks, stability of nonlinear systems.
\end{IEEEkeywords}

\IEEEpeerreviewmaketitle

\section{Introduction}\label{sec1}

\IEEEPARstart{F}{ormation} control of multi-agent systems has been inspired by collective animal behavior in nature, e.g. school of fish, formation of birds, pack of wolves, and more. The position-based, displacement-based, and distance-based controls are three general categories of formation control \cite{formation}. The distance-based control, in comparison with other methods, requires fewer measurements and higher interactions among the agents \cite{formation}. 
In distance-based formation control methods, interaction topology  is  usually  described by  graph rigidity or persistence,  while in displacement-based formation control the interaction topology is modelled by connectedness (graph Laplacian) \cite{formation}. Interactions can be modeled either by an undirected or directed graph. 

Attention to the distance-based control of double-integrator multi-agent systems is growing, as their applications are more common in comparison with single-integrator multi-agent systems. In \cite{de2019formation}, a distance-based control of single- and double-integrator multi-agent systems is studied. Similar to our work, in \cite{back}, authors propose the backstepping design for double-integrator multi-agent systems to address the formation control problem with topology switching. 

In comparison with existing distance-based methods such as \cite{sabaj,distance_leadfollow,oh2014distance}, which use single- or double-integrator dynamics for the system, the novelty of this paper is a study of a general class of nonlinear systems with unknown nonlinearity, time-delay on the dynamical states, and disturbance. The formation  control problem of nonlinear multi-agent systems with state time-delay function is studied in \cite{ma2015neural,chendelay,chen2018leader,chen}, where the authors consider a bounded time-delay for a nonlinear multi-agent system. In contrast with these methods where the Laplacian matrix is used to express the relation among the agents, our control law is based on rigid graph theory. While the formation control problems can be modeled with standard graph theory, a rigid graph theory offers modeling that minimizes the number of edges/distances that one needs to control in order to achieve the desired formation.

In this paper, the objectives are to keep agents, modeled by second-order nonlinear systems, in pre-specified distances from each other and to follow the target within a bounded trajectory. To achieve this, we propose a new control design for the nonlinear multi-agent systems relying on the backstepping control and rigid graph theory. The unknown nonlinear part of the system dynamics, as well as the disturbance, are approximated using NN, where an RBFNN (for more details see \cite{Lewisbook}) is used to estimate the unknown dynamics. To achieve the formation target tracking, the leader tracks a target within a bounded trajectory while the rest of the agents maintain the desired formation and follow the leader. 
 
The novelty in this work is a rigorous study and a proof of the multi-agent systems formation stability, when the agents are modeled with the second-order nonlinear dynamics and state time-delay on each agent. The main contributions of this paper are:

(i) The Lyapunov function is selected for the formation control problem and target tracking of second-order nonlinear multi-agent systems. To the best of the authors' knowledge, this paper is the first contribution proposing a distance-based formation for nonlinear multi-agent systems with state time-delay and disturbance. Compared with existing methods in \cite{chendelay,chen}, the dynamics of the system is not limited to a first-order nonlinear systems.

 (ii) Compared with the existing methods \cite{ma2015neural,chendelay,chen2018leader,chen}, which use the Laplacian matrix to express the relation among the agents, we use  rigid graph theory to represent those relations for a general class of nonlinear multi-agent systems. 
 
(iii) The minimum singular value of normalized rigidity matrix for infinitesimally and minimally rigid framework is utilized to design the neuro-adaptive controller using the backstepping technique to address the formation of nonlinear multi-agent systems. Also, an upper bound for the minimum singular value of the normalized rigidity matrix is introduced.
 
 
    
 \ifLongVersion   
    
\section{Preliminaries and Problem Formulation} \label{sec2}

\subsection{Preliminaries}
Interaction among the agents of the multi-agent systems is modeled by an undirected graph $G=(V,E)$, where $V=\{v_1,...,v_N\}$ specifies a set of vertices and $E\subset V \times V $ is its set of edges. With $|V|$ and $|E|$, we denote the number of vertices and edges, respectively. A pair of $(G,p)$ is known as a framework where $p=(p_1,...,p_N)$, $p_i\in \mathbb{R}^d$, such that $d\in\{2,3\}$, is the position of vertex $i$ in 2- or 3-dimensional space. More details about rigidity matrix can be found in \cite{sabaj,ECC20}.

The number of edges and vertices in 2D is related to rigidity through Laman's theorem.
\begin{theorem}[\hspace{-1.4mm}~\cite{laman}]\label{leman1}
A graph $G(V,E)$ in the plane is rigid, if and only if there exists a subgraph $G_1=(V,E_1)$ where $E_1\subset E$  with $|E_1|=2|V|-3$ in order that for any $V_1\subset V$, the associated induced subgraph $G_2=(V_1,E_2)$ of $G_1$ with $E_2\subset E_1$, satisfies $|E_2| \le 2|V_1|-3$. 
\end{theorem}

A graph is minimally rigid in 2D if and only if $|E| = 2|V|-3$; for more details please see \cite{mesbahi2010graph}.
\begin{lemma}[\hspace{-1.4mm}~\cite{mesbahi2010graph}]\label{lem11}
A framework in 2D, with $N>2$ is infinitesimally rigid, if and only if $rank(R_p)=2N-3$, where $R_p$ denotes the rigidity matrix of graph $G$.
\end{lemma}

\begin{lemma}[\hspace{-1.4mm}~\cite{de2019formation}]\label{lem1}
    For a vector $v_r \in R^2 $ and identity vector $\textbf{1}_N$, we have $R_p( \textbf{1}_N \otimes v_r)=\textbf{0}_{|E|}$.
\end{lemma}

Let us define the normalized rigidity matrix, $\bar{R}_p$, where each row of the rigidity matrix is divided by the two-norm of that row. Similar notion was first introduced in \cite{babazadeh2019distance}. Two properties of normalized rigidity matrix are given below.
   
\begin{lemma}\label{L2}
For infinitesimally and minimally rigid framework in 2D, the diagonal elements of $\bar{R}_p\bar{R}^T_p$ are equal to \textit{two}.
\end{lemma}

\begin{proof}
  Let  $\bar{R}_p(i,:)$ be an arbitrary row in normalized rigidity matrix as 
 {\small{ \begin{equation}\label{svd_rigid0_3}
 [0 ... 0,  \bar{R}_{2m-1,n},\  \bar{R}_{2m,n}, 0... 0, -\bar{R}_{2m-1,n},
      -\bar{R}_{2m,n}, 0 ...  0 ],
    \end{equation}}}with $\bar{R}_{2m-1,n}=(x_{mn})/(\sqrt{x^2_{mn}+y^2_{mn}})$ and $\bar{R}_{2m,n}=(y_{mn})/(\sqrt{x^2_{mn}+y^2_{mn}})$ being two consecutive elements of $2m-1$ and $2m$, respectively. Then, one can show that $\bar{R}_p(i,:)\bar{R}^T_p(i,:)= 2(\bar{R}^2_{2m-1,n}+\bar{R}^2_{2m,n}) =2$. \hfill $\blacksquare$ 
\end{proof}
 With  $\underline{\lambda}(\bar{R})$ and ${\lambda}_i(\bar{R})$, we denote the minimum and $i^{th}$ eigenvalue of $\bar{R}$, respectively.
 
\begin{theorem}\label{theor1}
Let the interaction among the agents in 2D be modeled by an undirected, infinitesimally and minimally rigid graph and the desired formation be any polygon with $N$ vertices. Then the minimum singular value of the normalized rigidity matrix is upper bounded by $\sqrt{2}$.
\end{theorem}
 
\begin{proof}
Consider a normalized rigidity matrix for a rigid graph with $N$ vertices. The normalized rigidity matrix has $2|N|-3$ rows (Lemma \ref{leman1}), diagonal terms in the matrix $\bar{R}=\bar{R}_p\bar{R}^T_p$ have a value of two (Lemma \ref{L2}), and $\bar{R}$ is full rank (Lemma \ref{lem11}). For any symmetric matrix $A$, one has $\sum \lambda_i (A) = tr(A)$. Therefore, $\sum \lambda_i (\bar{R}) =2(2N-3)$. As $\bar{R}$ is a positive definite, one has $ (2N-3) \underline{\lambda}( \bar{R})\le 2(2N-3)$, which yields $\underline{\sigma}( \bar{R}_p)  \le \sqrt{2}$ .  \hfill $\blacksquare$ 
\end{proof}

\begin{example}\label{square}
The minimum singular value of a normalized rigidity matrix for an infinitesimally and minimally rigid square framework in 2D (Fig. \ref{top1}) is $ \sqrt{2-\sqrt{2}}$.
\end{example}

\begin{figure}[t]
\hspace{.35cm}\includegraphics[scale=0.4]{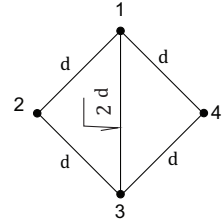} 
\caption{Formation and communication of an infinitesimally and minimally rigid framework with four agents in 2D.}\label{top1} 
\end{figure}

We assume that the multi-agent system communication graph is given by $G(V,E)$. The graph is assumed to be infinitesimally and minimally rigid, thus implying the connectivity \cite{jackson2007notes}.

\subsection{Problem Formulation}
Consider a second-order nonlinear multi-agent system consisting of $N$ agents where the dynamics of the $i$-th agent is given by   
\begin{equation}\label{eq1}
      \begin{split}
    \dot{p}_{i}=&v_{i} , \\
 \dot{v}_{i}=&f_{i}(p_{i},v_i)+g_{i}(p_{i},v_i)u_{i}(t)+\\&h_{i}(p_{i}(t-\tau_i),v_i(t-\tau_i)) 
  + w_{i}(p_i,v_i,t),   \ \ \ \
    \end{split}
          \end{equation}where the vectors $p_{i}\in\mathbb{R}^2$ and $v_{i}\in \mathbb{R}^2 $ represent the position and velocity of each agent respectively, $f_i(.)$ and $h_i(.) \in \mathbb{R}^2$ are the unknown smooth vector functions, considered to be continuously differentiable and locally Lipchitz and $\tau_i$ is an unknown time-delay,  $u_i\in \mathbb{R}^2$ is the control input, and $w_i(.) \in \mathbb{R}^2$ is a disturbance affecting each agent. Matrix $g_i(p_{i},v_i), \mathbb{R}^4 \rightarrow \mathbb{R}^{2 \times 2}$ is an unknown matrix. 

We establish standard assumptions \cite{lewis, chendelay,chen2018leader}, as follows.
 
 
\begin{assumption}\label{assump1}
    Unknown matrix $g_i(.)$ is either positive or negative definite, symmetric matrix with eigenvalues satisfying $0< \underline{g}_i\le  ||\lambda_1(g_i(.))|| \le  || \lambda_2(g_i(.))|| <\infty$, $i\in \{1,...,N\}$, and with $\underline{g}_i$ being a constant lower bound.
\end{assumption}
 
\begin{assumption} \label{assump3_1}
The vector function $h_i(.)$ is considered to be bounded, i.e., there exist a known function $\Upsilon_i$ such that $||h_i(x_i(t))||\le \Upsilon_i(x_i(t))$.
\end{assumption}

\begin{assumption} \label{assump2} 
Disturbance dynamics $w_i(x_i(t),t)$ is an unknown vector function that satisfies $||w_i(x_i,t)||\le\rho_i(x_i(t))$, where $\rho_i(x_i(t))$ is an unknown positive smooth function.
\end{assumption}  
 
\begin{assumption} \label{assump3_2}
Time-delay $\tau_i$ is unknown and bounded by $||\tau_i||\le \tau_M $, with a fixed bound $\tau_M$ for $i \in \{1, ... , N \}$.
\end{assumption}  
 The variable $x_i(t)$ will be defined later in the paper.   

In target tracking, we consider the \textit{first} agent to be the leader, while the remaining agents are followers. The control objectives are: (i) the leader tracks the target; (ii) the distance between neighboring agents $i$ and $j$ converges to desired distance $d_{ij}$:
\begin{equation}\label{eq15_2}
    \begin{array}{lr}
    ||p_i-p_j|| \rightarrow d_{ij} \ \ \ as \  t\rightarrow \infty, (i,j)\in E,
    \end{array}
\end{equation}where the $d_{ij}$s are positive and bounded by $\max(d_{ij})<D ,\forall (i,j) \in E$, with a fixed bound $D$.

Tracking error between the leader and the target is defined as $e_r= p_r-p_1$, where vector $e_{r}(t)=[e_{r1},e_{r2}]^T$ has two components $e_{r1}$ and $e_{r2}$ representing the leader's tracking error in $x$ and $y$ directions, respectively. 
Also, the vectors $p_r$ and $p_1$ $\in \mathbb{R}^2$ denote target and leader positions, respectively, and we define $v_r \triangleq \dot{p}_r $.
\begin{assumption}\label{assump7}
Time-function vectors $p_r$, $\dot{p}_r$, $\ddot{p}_r$ are bounded. The relative position and velocity of leader with respect to target, $p_r - p_1$, $\dot{p}_r- \dot{p}_1$, as well as the velocity of the target, $\dot{p}_r$, are known and can be broadcast to the followers \cite{de2019formation}. 
\end{assumption} 
Based on Assumption \ref{assump7}, let define a compact set such that  $\Omega_r= \{p_r,\dot{p}_r,\ddot{p}_r \ | \ ||p_r|| \le \bar{P}_r,||\dot{p}_r|| \le \bar{V}_r, ||\ddot{p}_r|| \le \bar{A}_r \}$, with fix bounds of $\bar{P}_r$, $\bar{V}_r$ and $\bar{A}_r$. 
\begin{remark}
The leader can estimate the position and velocity of the moving target using, for instance, radar technology \cite{sensor}. Also, it is able to broadcast the relative position and velocity to the followers as well as the target's velocity.
\end{remark}

\section{Formation Control of Second-Order Nonlinear Systems}\label{sec3}
\subsection{Control Algorithm Design}
The distance error for the multi-agent system (\ref{eq1}), modeled by an undirected graph, is given by $e_{ij}=||{p}_{ij}||-d_{ij}$ \cite{sabaj}. The distance error dynamics can be derived as 
    \begin{equation}\label{eq4}
    \begin{array}{lr} 
    \dot{e}_{ij}=\frac{{p}^T_{ij}(\dot{p}_{i}-\dot{p}_{j})}{||{p}^T_{ij}||}=\frac{{p}^T_{ij}(v_{i}-v_{j})}{e_{ij}+d_{ij}},
    \end{array}
\end{equation}
where $p_{ij}=p_i-p_j$. Let us define an energy function
   \begin{equation}\label{eq5}
    \begin{array}{lr} 
    M_1(e)= \frac{1}{2}\sum_{(i,j)\in E }^{} e^2_{ij}.
    \end{array}
\end{equation}Considering (\ref{eq4}) and taking a time-derivative of (\ref{eq5}), $\dot{M}_1$ is given by
   \begin{equation}\label{eq7}
    \begin{array}{lr}
    \dot{M}_1= \sum  \beta(e) \frac{{p}_{ij}^T    (v_i-v_j)} {||p_{ij}||}=\beta^T(e) \bar{R}_p \vec{x}_2,
    \end{array}
\end{equation}where $\beta(e)= (..., e_{ij},... )\in \mathbb{R}^{|E|}$ for $(i,j) \in E$. Moreover, $\vec{x}_1=[{p^T_1},...,{p^T_N}]^T$ and $\vec{x}_2=[{v^T_1},...,{v^T_N}]^T \in \mathbb{R}^{2N}$ is defined as an overall velocity vector for all agents.
    Using the backstepping technique \cite{sabaj}, we define $\textbf{s}=[ s^T_{1},  ...  , s^T_{N}]^T\in\mathbb{R}^{2N}$, $\vec{s}=\vec{x}_2-\nu$, 
where $\nu$ is an auxiliary variable given by 
\begin{equation}\label{eq10}
    \begin{array}{lr}
    \nu= u_s+\mathbf{1}_N\otimes (v_r+k_r e_r),
    \end{array}
\end{equation}
with $u_s=-k_v \bar{R}^T_p \beta (e)$ and $k_v$ being a positive constant.  

From (\ref{eq10}) one has
\begin{equation}
\begin{split}
    \nu_i&=-k_v \sum_{j\in N_i} (\frac{{p}_{ij} } {||p_{ij}||}e_{ij} )+(v_r+k_r e_r),  \\ 
     s_i&=v_i+k_v \sum_{j \in N_i}( \frac{{p}_{ij} } {||p_{ij}||}e_{ij}) -(v_r+k_r e_r).
    \end{split}
\end{equation}  
Achieving the desired formation in (\ref{eq10}) relies on $u_s$, while $\mathbf{1}_N\otimes (v_r+k_r e_r)$ is the term for tracking of the target by the leader and other agents. Let us define the potential function $V_1=M_1+M_2$, where $M_2=\frac{1}{2}\vec{s}^T\vec{s}=\frac{1}{2} \sum_{i=1}^{N} s_i^T s_i$. Taking time-derivative of $V_1$, using \textit{Lemma \ref{lem1}}, and equation (\ref{eq10}), one has
{\small{  \begin{equation}\label{eq12} 
        \begin{split} 
    \dot{V}_1 & = \beta^T(e) \bar{R}_p \vec{x}_2 + \vec{s}^T \dot{\vec{s}} \\
    & = \beta^T(e) \bar{R}_p \nu + \vec{s}^T [\dot{\vec{x}}_2 +\bar{R}_p^T\beta(e) -\dot{\nu}  ] \\
       &   = -k_v \beta^T(e)\bar{R}_p \bar{R}_p^T\beta(e)  + \sum_{i=1}^{N} s_i^T [f_i({x_i}) \\& +g_i({x_i})u_i+h_i({x_i}(t-\tau_i))+w_i({x_i},t)\\& +\sum_{j \in N_i}( \frac{{p}_{ij} } {||p_{ij}||}e_{ij}) -\dot{\nu}_i  ],
    \end{split}
\end{equation}}}where $x_i= [p_i^T,v_i^T]^T$. 
 
Using Assumptions \ref{assump3_1}-\ref{assump2} for (\ref{eq12}), and applying Cauchy's and Young's inequalities we have  
    \begin{equation}\label{eq12__1}
         \begin{split}
             s_i^T h_i({x}_i(t-\tau_i)) &\le \frac{||s_i||^2}{2}+\frac{\Upsilon_i^2({x}_i(t-\tau_i))}{2}, \\
             s_i^T w_i({x}_i,t)   &\le \frac{1}{2}+\frac{ ||s_i||^ 2\rho^2_i({x}_i(t))}{2}.
         \end{split}
     \end{equation}Substituting (\ref{eq12__1}) into (\ref{eq12}), one has  
{\small{\begin{equation}\label{eq26_1}
    \begin{split}
    \dot{V}_1 & \le -k_v \beta^T(e)\bar{R}_p \bar{R}_p^T\beta(e)  +\sum_{i=1}^{N} \bigg{(}s_i^T [f_i({x}_i) -\dot{\nu}_i\\ 
    & + (\sum_{j\in N_i}\frac{{p}_{ij}} {||p_{ij}||}e_{ij})   ]+\frac{ ||s_i||^2 \rho^2_i({x}_i(t))}{2} +\frac{||s_i||^2}{2} \\
    & +\frac{\Upsilon^2_i({x}_i(t-\tau_i))}{2}+  s_i^T g_i({x}_i)u_i\bigg{)}+\frac{N}{2}   .
    \end{split}
\end{equation}}}
To compensate for the unknown function $h_i(x_i(t-\tau_i))$, which is upper bounded by $\Upsilon_i(x_i)$, we add the following term to potential function $V_1$
 \begin{equation}\label{eqM_3}
    \begin{array}{lr}
M_3=\frac{1}{2}\sum_{i=1}^{N}\int_{t-\tau_i}^t \Upsilon_i^2(x_i(z)) dz,
    \end{array}
    \end{equation}
and its time-derivative is given by
 \begin{equation}\label{eq27_1}
    \begin{array}{lr}
    \dot{ M}_3= \frac{1}{2}\sum_{i=1}^{N}\big{(}\Upsilon_i^2({x}_i(t))-\Upsilon_i^2({x}_i(t-\tau_i))\big{)} .
    \end{array}
    \end{equation}Let us define potential function $V_2$ as  $V_2=V_1+ M_3$. 
      Taking time-derivative of $V_2$, and from inequality (\ref{eq26_1}) and equation (\ref{eq27_1}) yields  
 {{       \begin{equation}\label{eq30_1}
        \begin{split}
    \dot{V}_2 & \le -k_v \beta^T(e)\bar{R}_p \bar{R}_p^T\beta(e)  +\sum_{i=1}^{N} \bigg{(}s_i^T [f_i({x}_i) -\dot{\nu}_i\\ 
    & + (\sum_{j\in N_i}\frac{{p}_{ij}} {||p_{ij}||}e_{ij})   ]+\frac{ ||s_i||^2 \rho^2_i({x}_i(t))}{2} +\frac{||s_i||^2}{2} \\
    & +\frac{\Upsilon^2_i({x}_i(t-\tau_i))}{2}+  s_i^T g_i({x}_i)u_i\bigg{)}+\frac{N}{2}   +\\ & \frac{1}{2}\sum_{i=1}^{N}\big{(}\Upsilon_i^2({x}_i(t))-\Upsilon_i^2({x}_i(t-\tau_i))\big{)}  \\
    &  \le -k_v \beta^T(e)\bar{R}_p \bar{R}_p^T\beta(e)  +\sum_{i=1}^{N} \bigg{(}s_i^T [f_i({x}_i) -\dot{\nu}_i\\ 
    & + (\sum_{j\in N_i}\frac{{p}_{ij}} {||p_{ij}||}e_{ij})   ]+\frac{ ||s_i||^2 \rho^2_i({x}_i(t))}{2} +\frac{||s_i||^2}{2} \\
    & +  s_i^T g_i({x}_i)u_i\bigg{)}+\frac{N}{2}  +\frac{1}{2}\sum_{i=1}^{N}\big{(}\Upsilon_i^2({x}_i(t))\big{)}.
    \end{split}
    \end{equation}}}with $T_i(s_i,{x}_i)=f_i({x}_i)+\frac{1}{2}s_i \rho_i({x}_i(t))$. Motivated by \cite{chendelay}, let $\{p_i,v_i\} \in \Omega_{x_i}$ be a compact set, then $\Omega_{\varphi_i} \subset \Omega_{x_i}$, where  $\Omega_{\varphi_i} = \{s_i \ | \ ||s_i|| < o_i \}$  with $o_i$ chosen as a small arbitrary constant. As $\Omega_{\varphi_i}$ is an open set, its complement set, $\Omega^0_{\varphi_i} =\Omega_{x_i}-\Omega_{\varphi_i}$, is a compact set.
   
   The ideal approximation of $T_i(x_i,s_i)$ over the compact set $\Omega^0_{\varphi_i}$ is given by $ T_i(s_i,{x}_i)={W}^T_i\varphi_i(s_i,{x}_i) +\epsilon_i(s_i,{x}_i)$, where $W_i\in \mathbb{R}^{\eta_i\times 2}$, $\varphi_i(s_i,{x}_i) \in \mathbb{R}^{\eta_i}$, and $\eta_i$ are the ideal NN weights matrix, activation function and number of neurons, respectively. The approximation through RBFNN over a compact set is $\hat{T}_i(s_i,{x}_i)=\hat{W}^T_i\varphi_i(s_i,{x}_i)$ where $\hat{W}_i\in \mathbb{R}^{\eta_i\times 2}$. Variable $\tilde{W}_i=W_i-\hat{W}_i$ is the estimation error of NN weights matrix. Based on NN approximation property and assuming $T_i$ to be continuously differentiable function, there exists a sufficient number of neurons $\eta_i^*$ such that if $\eta_i>\eta_i^*$, the NN approximation error $\epsilon_i$ is bounded by $||\epsilon_i||\le \xi_i$.

\section{Main Result}
 We propose the following control law for the multi-agent system (\ref{eq1}) 
\begin{equation}\label{eq15_26}
          u_i =
       \left\{
    \begin{array}{ll}
       -c_i(t) s_i -\frac{1}{\underline{g}_i} (\hat{W}^T_i\varphi_i(s_i,{x}_i)\\-\sum_{j \in N_i} (\frac{{p}_{ij}}{||p_{ij}||}e_{ij})+\gamma_i) \\ -(\frac{1}{2 {\underline{g}_i}})s_i^{-1} \Upsilon^2_i({x_i}(t)) , & s_i\in \Omega^0_{\varphi_i} \\
        \vec{0}, & s_i\in \Omega_{\varphi_i}
    \end{array}
\right.
\end{equation}
where $\hat{W}_i$ is the current estimation of ideal weight $W_i$, $s_i^{-1}= s_i/||s_i||^2$, control gain $c_i(t)>0$, and $\gamma_i$ is given by
     \begin{equation}\label{eq15-1}
     \begin{split}
              \gamma_i= & -k_v\sum_{j \in N_i} \frac{({v}_{ij}e_{ij}+\dot{e}_{ij} p_{ij})||p_{ij}||^2- p^T_{ij} v_{ij} p_{ij} e_{ij} } {||p_{ij}||^3}  \\ &-b\ sgn(s_i)+ k_r \dot{e}_r,
        \end{split} 
    \end{equation}with $b\ge \sqrt{2N} ||\dot{v}_r||_{\infty}$. Agents' NN tuning law is given by
   \begin{equation}\label{eq15_27}
    \begin{array}{lr}
    \dot{\tilde{W}}_i=- F_i \varphi_i(x_i,s_i) s_{i}^T +\kappa_i F_i \hat{W}_i,
    \end{array}
\end{equation}  where $\kappa_i>0$ is a constant and $F_i=\Pi_i I_{\eta_i}$ is a positive definite matrix with $\Pi_i$ is a positive constant, $I_{\eta_i}$ is the $\eta_i \times \eta_i $ identity matrix and ${W}=diag({W}_i)$.

\subsection{Stability Analysis}
In this section we formalize the proposed control law (\ref{eq15_26}) with a rigorous stability result. The next theorem provides a result that guarantees that the leader follows the target and followers maintain desired formation with the leader.
\begin{theorem}  
Let the required framework be modeled as an undirected, infinitesimally and minimally rigid graph. Under Assumptions 1-5, select the control input as (\ref{eq15_26}), $b\ge \sqrt{2N} ||\dot{v}_r||_{\infty}$, adaptive NN weights tuning law as (\ref{eq15_27}), and the control gain $c_i(t)$ as
\begin{equation}\label{eqC}
    c_i(t)=\frac{\Gamma_i}{\underline{g}_i} \big{(} \frac{1}{2 ||s_i||^2} \int_{t-\tau_M}^t \Upsilon_i^2(x(z))dz + 1 +\frac{k_c}{2\Gamma_i} \big{)},
\end{equation}where $-k_v   \underline{\sigma} (\bar{R}_p) \le
-\frac{\underline{\Gamma}}{2}$,  $\kappa_i \ge \underline\Gamma \Pi^{-1}_i   $, and  $k_c\ge \underline\Gamma$, $\underline{\Gamma}=min\{\Gamma_1,..., \Gamma_N\}$. Then, the inter-agent distance errors and NN weights matrix estimation errors are UUB for initial conditions that belong to the compact set ${\Omega}_{0}$:
\begin{equation}\label{bound}
\begin{split}
    \Omega_{0}= & \{ \vec{x}_1(0),\vec{x}_2(0), \hat{W}(0), p_r(0)|  p_r(0) \in \Omega_r,\\&  \divides||p_{ij}||-d_{ij} \divides \le \frac{\sqrt{\delta}}{|E|}, \text{and} \ p_i\ne p_j, (i,j) \in E\}, \\ 
\end{split}
 \end{equation}
where $\delta$ is a small positive constant.
\end{theorem}
 
\begin{proof} 
The error dynamics with respect to $e_{ij}$ and $s_i$ in a closed loop system with (\ref{eq15_26})-(\ref{eq15_27}) has right-hand side discontinuity because of $sgn(s_i)$ in (\ref{eq15-1}). We choose the non-smooth Lyapunov function candidate $V=V_2+ M_4,$ with $ M_4=\frac{1}{2}\sum_{i=1}^{N} tr(\tilde{W}_i^TF_i^{-1}\tilde{W}_i)$.  Let $\dot{\varsigma}_i=\mathcal{F}_i(\varsigma_i,t)$ be the closed loop system where $\varsigma_i=[e_{ij},s_i]$, then  $\mathcal{F}_i(\varsigma_i,t)$ is continuous everywhere except in the set $\{(\varsigma_i,t) \ | \ ||s_i|| < o_i \}$. For this system Filippov solution exists by satisfying differential inclusion $\dot{\varsigma}_i\in {K}_i {[\mathcal{F}}_i] ({\varsigma}_i,t)$ where ${K}_i {[\mathcal{F}}_i] ({\varsigma}_i,t)$ is an upper semi-continuous, nonempty, set-valued map \cite[p.~171]{de2019formation}. Then time-derivative of $V$ is given by 
 { \small{ \begin{equation}\label{eq48_11}
\begin{split}  
     \dot{V}  &      \stackrel{a.e.}{\in} \sum_{i=1}^N\frac{\partial V}{\partial \varsigma_i}  {K}_i {[\mathcal{F}}_i] ({\varsigma}_i,t)  
           \\ & \subset  -k_v \bar{R}_p\beta^T(e)\beta(e)\bar{R}^T_p  +\sum_{i=1}^{N} \bigg{(}s_i^T T_i(s_i,{x}_i)
    +\frac{\Upsilon_i^2({x}_i(t))}{2} \\  
    & +\frac{||s_i||^2}{2}+ s_i^T \Big{(}g_i({x}_i)\big{[}-c_i(t) s_i -{\frac{1} {\underline{g}_i}} \hat{T}_i(s_i,{x}_i) \\
   &+(\frac{1} {\underline{g}_i})(-\sum_{j \in N_i}\frac{{p}_{ij}} {||p_{ij}||}e_{ij}+\gamma_i)  -(\frac{1}{2 {\underline{g}_i}})s^{-1} \Upsilon_i^2({x}_i(t))\big{]}
    \\ & + \sum_{j \in N_i}\frac{{p}_{ij}} {||p_{ij}||}e_{ij}-\dot{\nu}_i\Big{)} \bigg{)}+ \frac{N}{2}+\sum_{i=1}^{N} tr(\tilde{W}_i^TF_i^{-1}\dot{\tilde{W}}_i).
    \end{split} 
\end{equation}}}
Applying Cauchy's inequality and substituting equation (\ref{eq15-1}), one has
   \begin{equation}\label{eq48_1}
    \begin{split}
       \dot{V}   &  \le -k_v \bar{R}_p\beta^T(e)\beta(e)\bar{R}^T_p  +\sum_{i=1}^{N} \bigg{(}s_i^T T_i(s_i,{x}_i)
    +\frac{\Upsilon_i^2({x}_i(t))}{2} \\  
    & +\frac{||s_i||^2}{2}+ s_i^T \Big{(}g_i({x}_i)\big{[}-c_i(t) s_i -{\frac{1} {\underline{g}_i}} \hat{T}_i(s_i,{x}_i)\\ 
    &+(\frac{1} {\underline{g}_i})(-\sum_{j \in N_i}\frac{{p}_{ij}} {||p_{ij}||}e_{ij}+\gamma_i)  -(\frac{1}{2 {\underline{g}_i}})s^{-1} \Upsilon_i^2({x}_i(t))\big{]}  \\ & +\sum_{j \in N_i}\frac{{p}_{ij}} {||p_{ij}||}e_{ij}-\dot{\nu}_i\Big{)} \bigg{)}+ \frac{N}{2}+\sum_{i=1}^{N} tr(\tilde{W}_i^TF_i^{-1}\dot{\tilde{W}}_i)\\ &
     \le -k_v \bar{R}_p\beta^T(e)\beta(e)\bar{R}^T_p  +\sum_{i=1}^{N} \bigg{(}
    +\frac{\Upsilon_i^2({x}_i(t))}{2} \\  
    & +\frac{||s_i||^2}{2}+ s_i^T \Big{(}g_i({x}_i)\big{[}-c_i(t) s_i \\ 
    &+(\frac{1} {\underline{g}_i})(-\sum_{j \in N_i}\frac{{p}_{ij}} {||p_{ij}||}e_{ij}+\gamma_i)  -(\frac{1}{2 {\underline{g}_i}})s^{-1} \Upsilon_i^2({x}_i(t))\big{]}  \\ & +\sum_{j \in N_i}\frac{{p}_{ij}} {||p_{ij}||}e_{ij}-\dot{\nu}_i\Big{)} \bigg{)}+ \frac{N}{2}+\sum_{i=1}^{N} tr(\tilde{W}_i^TF_i^{-1}\dot{\tilde{W}}_i)\\ & +
      \sum_{i=1}^{N} s^T_i(\hat{T}_i(s_i,{x}_i)-T_i(s_i,{x}_i)),
    \end{split}
\end{equation}
where $\tilde{T}_i=\hat{T}_i(s_i,{x}_i)-T_i(s_i,{x}_i)=\tilde{W}_i^T \varphi_i(s_i,x_i)$. As A=$s_i^T \tilde{T}_i $ is a scalar we have $A=A^T$; therefore, we can write the following equation
 \begin{equation}\label{eq482}
    \begin{split}
      \sum_{i=1}^{N} s^T_i(\hat{T}_i(s_i,{x}_i)-T_i(s_i,{x}_i))&  = \sum_{i=1}^{N} s^T_i(\tilde{W}_i^T \varphi_i(s_i,x_i) ) \\ & = \sum_{i=1}^{N} tr(\tilde{{W}}_i^T \varphi_i(s_i,{x}_i) s^T_i )).
    \end{split}
\end{equation}
Consequently,
\begin{equation}
\begin{split}
 &\sum_{i=1}^{N} tr(\tilde{W}_i^TF_i^{-1}\dot{\tilde{W}}_i)+  \sum_{i=1}^{N} s^T_i(\hat{T}_i(s_i,{x}_i)-T_i(s_i,{x}_i))= \\ & \sum_{i=1}^{N} tr\big{(}\tilde{{W}}_i^T (F_i^{-1}\dot{\tilde{W}}_i+ \varphi_i(s_i,{x}_i ) s^T_i ) \big{)}.
\end{split}
\end{equation}
which yields

    {{   \begin{equation}\label{eq49_1}
    \begin{split}
    \dot{V}   & \le -k_v \underline{\sigma}(\bar{R}_p )\beta^T(e)\beta(e)  +\sum_{i=1}^{N} \bigg{(}||s_i||~||\epsilon_i(s_i,{x}_i)||  \\
    & +\frac{||s_i||^2}{2}+ s_i^T \Big{(}-g_i({x}_i)c_i(t) s_i -(b\  sgn(s_i) +  \dot{v}_r)\Big{)} \bigg{)}  \\ 
    & +  \frac{N}{2}+  \sum_{i=1}^{N} tr(\tilde{{W}}_i^T\big{(}F_i^{-1}\dot{\tilde{W}}_i+ \varphi_i(s_i,{x}_i \big{)} s^T_i )).
      \end{split}
    \end{equation}}}
  By utilizing Young's inequality, substituting equations (\ref{eq15_27}), (\ref{eqC}) into (\ref{eq49_1}), and  using Young's inequality and assuming boundedness of ideal NN weights matrix by $||W||_F \le W_M$, with a fixed bound $W_M$, we have 
\begin{equation}\label{eq54_1}
 \begin{split}
    \dot{V}  &
    \le  -\underline{\Gamma}M_1 -\underline{\Gamma}M_2  -\underline{\Gamma}M_3 - \underline{\Gamma} \sum_{i=1}^{N}\frac{\Pi^{-1}_i}{2}||\tilde{W}_i||^2_F \\
    & + \frac{N}{2} + \frac{1}{2}\sum_{i=1}^{N} \Big{(} \xi^2_i+ \kappa_i W^2_M \Big{)}   \le  -\underline{\Gamma} V + B_M,
     \end{split}
    \end{equation}
    where $B_M= \frac{N}{2} + \frac{1}{2}\sum_{i=1}^{N} \Big{(} \xi^2_i+ \kappa_i W^2_M \Big{)} $. From \cite[Lemma~1.1]{ssg}, inequality (\ref{eq54_1}) implies  
  \begin{equation}\label{eq55_1}
V(t) \le    \frac{B_M}{\underline{\Gamma}}+(V(0)+\frac{B_M}{\underline{\Gamma}})e^{-\underline{\Gamma}t}.    
\end{equation}
 It can be seen that using control law (\ref{eq15_26}), the distance errors are uniformly ultimately bounded (UBB).  If $s_i\in \Omega_{\varphi_i}$, as  ${o_i}$ is chosen small enough, it follows that the formation control has been achieved and no more control effort is required. \hfill $\blacksquare$ 
\end{proof}
\begin{dis}
If inequality (\ref{eq55_1}) holds, as $t \rightarrow \infty$ the radius can be reduced by choosing $\underline{\Gamma}$ large enough. Thus, $\underline{\Gamma}$ has a direct impact on formation stability and tracking performance.
\end{dis}
\begin{remark}
Through control law (\ref{eq15_26}), the multi-agent systems moves toward the equilibrium which is located in $\Omega_{0}$; as a result the distance error remains in the invariant set  ${\Omega}_{0}$ \cite{de2019formation}. Moreover, $\delta$ is chosen sufficiently small positive constant and consequently, $p_{ij}\ne0$.

\end{remark}

\begin{lemma}\label{lem3}
Consider the following function
{\small{\begin{equation}\label{NEEW}
\begin{split}
   V&= \frac{1}{2} \beta^T(e)\beta(e)+   \frac{1}{2} tr(\tilde{W}^T F^{-1}\tilde{W}) + \frac{1}{2} \vec{s}^T \vec{s} \\ & + \frac{1}{2}\sum_{i=1}^{N}\int_{t-\tau_i}^t \Upsilon_i^2(x_i(z)) dz,
    \end{split}
\end{equation}}} where $\tilde{W}=diag(\tilde{W}_i)$. If $\dot{V}$ satisfies $\dot{V} \le - \alpha_1 V + \alpha_2$, then for a bounded initial conditions in a bounded set $\Omega_0$ (\ref{bound})\\
i) the states and NN weights remain within a bounded set
\begin{equation}\label{eq63}
\begin{split}
   \Omega_{b_i}=& \{p_i(t), v_i(t), \hat{W}_i(t)| ~ ||p_i(t)||\le \bar{P}_i^*, ||v_i(t)||\le \bar{V},\\ &
     ||\hat{W}_i(t)||_F \le W_M+ \bar{W}_i, p_r(t) \in \Omega_r     \},
   \end{split}
\end{equation} where constants $\bar{P}_i^*$ and $\bar{W}_i$ and $\bar{V}$  are defined as  
\begin{equation}
    \begin{split}
        \bar{P}_i^*&=  (\alpha^*_i+1) \bar{P} +\bar{P}_r + \bar{\zeta}, \ 
        \bar{W}_i= \sqrt{\frac{2V(0)+\frac{2\alpha_2}{\alpha_1}}{\Pi_i}},\\
        \bar{V}&= \bar{S}(1+k_v \bar{\sigma}(\bar{R}_p))+\sqrt{2N}\big{(}||\dot{p}_r|| + \bar{\zeta} \ \big{)},
    \end{split}
\end{equation}
and $\alpha^*_i$ denotes the number of vertices in the minimum path from the leader to the $i$-th agent and
\begin{equation}
    \begin{split}
        \bar{S}&=\sqrt{2V(0)+ \frac{2\alpha_1}{\alpha_2} },  \bar{P} = \sqrt{2V(0)+\frac{2\alpha_2}{\alpha_1}}+D, \\ 
             \bar{\zeta}&= ||e_r(0)|| +\frac{\sqrt{2}\bar{S}}{k_r}(1+k_v \bar{\sigma}(R_p)).
    \end{split}
\end{equation}
ii) The states and weights converge  to a compact set
\begin{equation}\label{eq65}
   \Omega_{c_i}= \{p_i(t), \hat{W}_i(t)| \lim_{t \rightarrow \infty} ||e_{ij}||=\Xi_{e_i} ,   \lim_{t \rightarrow \infty} ||\tilde{W}_{i}||=\Xi_{W_i} \}, 
\end{equation}
with $\Xi_{e_i}=\sqrt{\frac{2\alpha_2}{\alpha_1}}$ and $\Xi_{W_i}=\sqrt{\frac{2\alpha_2}{\alpha_1 \Pi_i}}$.
\end{lemma}
\vspace{2pt}
\begin{proof}
Similar proof is provided in \cite{ssg}, hence we skip the proof here.
\hfill $\blacksquare$ 
\end{proof}

It can be shown, using Gershgorin circle theorem, that $\bar{\sigma}(\bar{R}_p) \leq \sqrt{2N-2}$. Thus, the compact set of $\Omega^0_{\varphi_i}$ defined in (\ref{eq15_26}) can be specified as
\begin{equation}\label{eq80}
\Omega^0_{\varphi_i} = \{ p_i,v_i, s_i \ | \ ||p_{i}|| \le \bar{P}_i^*, ||v_i|| \le \bar{V}, ||s_i|| \le \bar{S} \}.
\end{equation}
\begin{remark}
The variables $s_i$ and $\gamma_i$ are bounded as all of their components are bounded as well as the control gain $c_i(t)$ and the control input (\ref{eq15_26}).
\end{remark}

\section{Simulation Results}\label{sec5}
Here we provide numerical results in order to verify and validate the performance of the proposed control method. In addition, two comparisons with an existing displacement-based control are conducted to demonstrate advantages of the proposed method. In order to evaluate and quantify the performance of each method, a performance index is introduced. Moreover, based on Example \ref{square}, we also compared the performance of the proposed method with an existing method presented in \cite{chen2018leader}.
\begin{figure}[t]
        \centering
        \includegraphics[scale=.48]{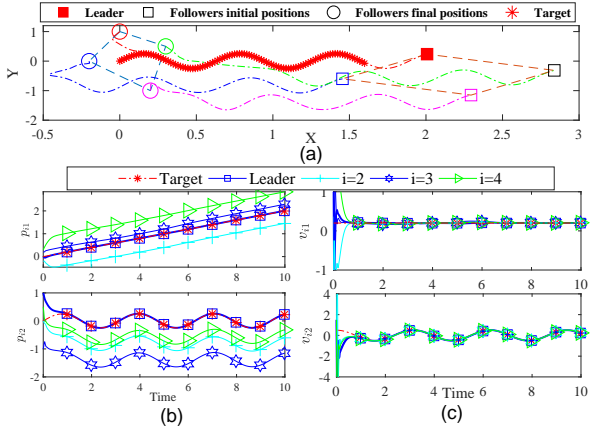}
        \caption{(a) Desired formation is a square with side $d=1$; (b) Trajectories of the agents and the target along $x$ and $y$ directions; (c) Velocity of agents and the target along $x$ and $y$ directions.}
        \label{fig1_1}
\end{figure} 
    
\begin{figure}[t]
        \centering
        \includegraphics[scale=.54]{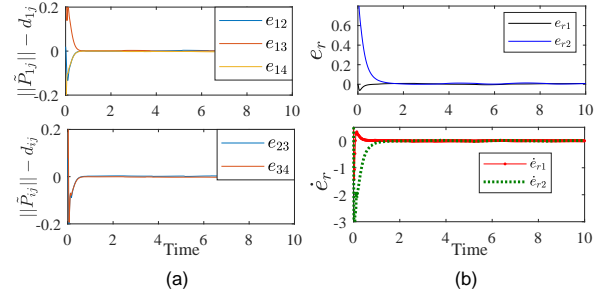}
        \caption{(a) The distance errors $e_{ij}, \ (i,j) \in E$ of proposed method. (b) Target  tracking  error  of  the  leader  and  its  time-derivative  of  proposed  method.}
        \label{fig7}
\end{figure} 

\begin{figure}[t]
        \centering
        \includegraphics[scale=.61]{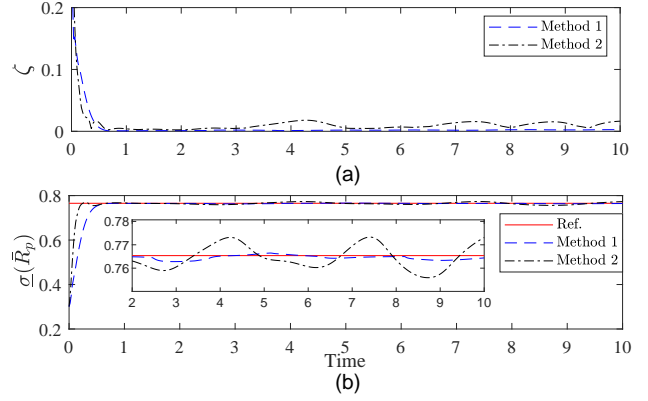}
        \caption{(a) Performance comparison (ADE) between proposed distance-based method (Method 1) and displacement method (Method 2) \cite{chen2018leader}; (b) Minimum singular value of normalized rigidity matrix comparison between proposed distance-based method (Method 1) and displacement method (Method 2) \cite{chen2018leader} (solid line is the desired value for minimum singular value of normalized rigidity matrix in Example \ref{square}).}
        \label{fig1-comp}
\end{figure}

\begin{example}\label{ex2}
Consider a nonlinear multi-agent system with four agents in a plane, where dynamics of each agent is given by (\ref{eq1}):
{\footnotesize{
\begin{equation}\label{exam-1}
    \begin{split}
        \dot{p}_{i}(t) =&v_{i}(t), \\
        \dot{v}_{i}(t) =&\left(
        \begin{array}{c}
            a_{i1}v_{i2}(t)v_{i1}(t)+sin(a_{i1} p_{i1}(t))  \\
            b_{i1}p_{i2}(t)v_{i2}(t)+cos(b_{i1} p_{i2}(t))  
        \end{array} \right) \\ 
        +  &\left(
        \begin{array}{cc}
            1+cos(v_{i4})sin(v_{i3}^2) & 0  \\
            0 & 1+cos(v_{i3})sin(v_{i4}^2) 
        \end{array} \right) u_i(t) \\
        +&\left(\begin{array}{c}
         c_{i1} p_{i1} (t-\tau_i) cos(v_{i1} (t-\tau_i)) \\
          c_{i2} p_{i2} (t-\tau_i) sin(v_{i2} (t-\tau_i))
        \end{array} \right) \\
        +&\left(\begin{array}{c}
        d_{i1}v_{i2}p_{i1}^2cos(1.5t) \\
        d_{i2}(v_{i1}+p_{i2})sin(t)
        \end{array} \right),
        \ \ \ \ \     i=1,...,4.
    \end{split}
\end{equation}}} Parameters $a_{i1}$, $b_{i1}$, $c_{i1}$, $ c_{i2}$, $d_{i1}$, $d_{i2}$ of each agent are given in Table \ref{tab2}. 
The initial conditions for the four agents are: $p_{1}(0)=(0,1)^T$, $v_1(0)=(1,1.5)^T$, $p_{2}(0)=(-0.2,0 )^T$ ,$v_2(0)=(-1,1)^T$, $p_{3}(0)=(0.2,-1)^T$, $v_3(0)=(1,-1)^T$, and $p_{4}(0)=(0.3,0.5)^T$, $v_4(0)=(0.5,0.5)^T$.  The time-delays are $\tau_1=0.10$, $\tau_2=0.18$, $\tau_3=0.13$, $\tau_4=0.12$, and $\tau_M=0.2$. 
The target velocity and initial position are $v_r=[0.2,0.5cos(2t)]^T$ and $p_r(0)=[0,0]^T$, respectively. To select $||\dot{v}_r||_{\infty}=0.5$, we choose $b=3$ and $k_r=3$. RBFNN is selected with 9 neurons and $\kappa_i =2.5, F_i=10I_{9}, k_v=15$. To satisfy Assumptions \ref{assump3_1} and \ref{assump2}, we choose $\Upsilon_i(x_i(t))=\sqrt{(c_{i1}p_{i1})^2+(c_{i2}p_{i2})^2}$  and $\Gamma_i=2$,  $k_C=200$ in (\ref{eqC}).

The desired distances and communication topology are given in Fig. \ref{top1}, with $d=1$. Figs.  \ref{fig1_1}-\ref{fig7} show results of the proposed control law (\ref{eq15_26}) for a nonlinear multi-agent system. Fig. \ref{fig1_1}(a) shows how agents form a square where each side is equal to the desired distance ($d=1$). Fig. \ref{fig1_1}(b) represents the trajectories of the agents and the target in $x$ and $y$ directions. Velocity of agents and the target in $x$ and $y$ directions are shown in Fig. \ref{fig1_1}(c).
 
To compare the results with a displacement-based method, the modified version of \cite{chen2018leader} is simulated. The method in \cite{chen2018leader} shows a leader-following consensus control of second-order  nonlinear systems with state time-delay. By applying this method and adding constant displacements of the desired positions as \cite[p.~127]{mesbahi2010graph} to the formation control problem, we obtain a shape-based control.

To evaluate the performance of proposed method in comparison with the displacement method of \cite{chen2018leader}, we define an average distance error (ADE) $\zeta(t)$ as $\zeta (t)=\frac{1}{|E|}\sum_{(i,j) \in E} |(||\tilde{p}_{ij}||-d_{ij})|.$

We use ADE, to compare our results with \cite{chen2018leader}. As shown in Fig. \ref{fig1-comp}(a), the proposed method improves the ADE when applied to a second-order nonlinear multi-agent systems. Moreover, to provide a better comparison, minimum singular values of normalized rigidity matrix of square for our proposed method (dashed line) and Method 2 (dash-dotted line) \cite{chen2018leader} are depicted in Fig. \ref{fig1-comp}. (b). It has been shown in Example \ref{square}, if the formation reaches the square shape, the minimum singular value is $\sqrt{2-\sqrt{2}}$. This shows an improvement in performance in comparison with \cite{chen2018leader}. Moreover, it can be noted that the minimum singular value of normalized rigidity matrix is always less than the introduced upper bound in Theorem \ref{theor1}.
\vskip 6pt
 \begin{table}[htbp]
                \begin{center}
                \caption{Parameters $a_{i1}, b_{i2}$, $c_{i1}, c_{i2}$, $d_{i1}, d_{i2}$ for the $i$-th agent. }
\vspace{-.1cm}      {\small{     \begin{tabular}{c|c c c c c c c}
                \hline 
                \textbf{$i$} & $a_{i1}$ &  $a_{i2}$ & $b_{i1}$ &$b_{i2}$ & $c_{i1}$ &$c_{i2}$    \\
                \hline
                1& 0.3 &  1 & 1 & -1 & -2.4 & 2.1  \\
                2&  0.7  &  -0.2 &  -1.2 & -2.2 & 1.8 & -1.5   \\
               3&   -0.7   &  -0.8 & 2.1& 1.2 & -0.4 & 1.3 \\
                 4&  -0.6   & 0.4 & -0.5 & -0.7 & 0.6& 0.8 \\  
                \hline
                \hline
             \end{tabular}}}
            \label{tab2}
        \end{center}
    \end{table}
  \end{example}
\section{Conclusion}
A neuro-adaptive backstepping and rigid graph theory-based control has been proposed for a distance-based formation control and target tracking of second-order nonlinear multi-agent systems modelled by an undirected graph in the presence of bounded disturbance and unknown state time-delay. The RBFNN has been used to compensate for the unknown nonlinearity of dynamical system and disturbance. The rigorous stability analysis based on the Lyapunov stability theory shows UUB of distance error. The upper bound for the minimum singular value of the normalized rigidity matrix has been introduced and used in designing of the control systems. The simulation results have verified the performance of the proposed method. Two sets of comparisons have been provided to demonstrate the efficiency and improvements of the proposed method compared with the recent results in the literature. Future work will consider time-delay in communication links of the undirected graph.
\ifLongVersion


\bibliographystyle{IEEEtran}
\bibliography{Journal}


\end{document}